\documentclass{article}
\usepackage{float}
\usepackage{authblk}

\title{Blind Quantum Computation Using a Circuit-Based Quantum Computer}
\author[1]{Yuichi Sano\thanks{\texttt{sano.yuichi.77v@st.kyoto-u.ac.jp}}} 
\affil[1]{Department of Nuclear Engineering, Kyoto University, Nishikyo-ku, Kyoto 615-8540, Japan} 

\date{\today}

\usepackage{graphicx}
\usepackage{amsmath,amsthm,amssymb}
\usepackage{amsfonts}
\usepackage{siunitx}

\usepackage{qcircuit}
\usepackage{braket}
\usepackage{cite}
\usepackage{comment}

\theoremstyle{definition}
\newtheorem{thm}{Theorem}
\newtheorem{defn}[thm]{Definition}

\begin{document}

\maketitle

\begin{abstract}
When a universal quantum computer is used by the public, it is assumed that it will be in the form of a  quantum cloud server that exists in a few bases due to its cost.
In this cloud server, privacy will be a crucial issue, and a blind quantum computation protocol will be necessary so that each user can use the server without the details of the calculations being revealed.
It is also important to be able to verify that the server is performing calculations as instructed by the user, since quantum calculations cannot be verified by classical computation.
In  this paper, we put forward a protocol that achieves blindness using the quantum one-time pad for  encryption and a T-like gate, and while verifying computation using trap qubits.
\end{abstract}

\section{Introduction}
When universal quantum computer is used by public, it is assumed that it will be used as a quantum cloud server that exists in a few bases because the quantum computer is expensive.
In this cloud server, privacy will be a crucial issue.
Thus, blind quantum computation  (BQC) protocol is needed so that each user can use the server without revealing the details of his or her calculations.
\cite{Previous research1,Childs,Previous research2,BFK,MF protocol,Broadbent,Morimae hayashi,half,Full}.

In \cite{Childs}, such a method was proposed based on the quantum onetime pad.
Similar to classical one-time pads [10], the quantum one-time pad uses the encryption  key only once, and the server cannot learn anything about the  user's quantum state.
However, this protocol needs multiple two-way quantum communications. 
In addition, the user is required to have a quantum memory on which a SWAP gate is executed.
In \cite{Broadbent}, another protocol was proposed that  requires neither quantum gates and two-way quantum communication nor  quantum memory and SWAP gates during its computation.
In this  protocol, however, while the input and output are encrypted, the calculation process is revealed to the server. 
This is a crucial drawback because an algorithm itself can constitute  important information that should be kept secret. 
In addition, as a malicious server might have performed a calculation different from the user’'s instruction, a user must have anthe ability to verify the calculation\cite{BFK, MF protocol, Morimae hayashi, FK}.
However, quantum computers generally cannot be simulated in polynomial time by classical  computers, and a user with the limited ability assumed in previous research cannot calculate  whether the results obtained from the server are correct.
Furthermore, such a user cannot verify calculations using trap qubits, which are closely related to  the secrecy of the calculation process.

In this paper, we propose a novel BQC protocol using rotation gates in addition  to the quantum one-time pad.
Our protocol enables verification by trap qubits and can be extended to fault-tolerant computation.
In this protocol, the abilities  required for the user are equivalent to those required in the previous BQC protocol\cite{BFK}.

\section{Preliminaries}
In this section, we describe gate teleportation and the encryption method as known the quantum one-time pad, which are both used in the proposed protocol.
See \cite{Nielsen-Chuang} for the general quantum computation notation.

\subsection{Gate teleportation}
We explain gate teleportation for a $T$ gate that is used for universal gate sets and an $A_\theta$ gate that is used for blindness in the protocol. The $T$ gate is
\begin{equation*}
T=\begin{pmatrix}1 & 0\\ 0 & e^{\frac{i\pi}{4}}\end{pmatrix},
\end{equation*}
and the $A_\theta$ gate is
\begin{equation*}
A_\theta=\begin{pmatrix}1 & 0\\ 0 & e^{i\theta}\end{pmatrix}.
\end{equation*}
When $\theta=\frac{\pi}{4}$, the $A_\theta$ gate is equivalent to the $T$ gate.

For a given state $\ket{\psi}$, $A_\theta \ket{\psi}$ is obtained by using gate teleportation (Figure \ref{fig:Gate teleportation}), without directly executing the $A_\theta$ gate, where $a$ is the measurement result, and $\ket{A_\theta}$ is 
\begin{equation*}
    \ket{A_\theta}=\frac{1}{\sqrt{2}}(\ket{0}+e^{i\theta}\ket{1}).
\end{equation*}

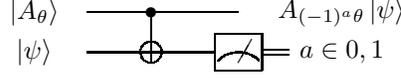
\begin{figure}
\begin{equation*}
    \Qcircuit @C=1em @R=.7em {
    \ket{A_\theta} & & & \qw & \ctrl{1} & \qw & \qw & & & A_{(-1)^a\theta}\ket{\psi}\\
    \ket{\psi}&&& \qw & \targ  & \qw & \meter & \cw & & a \in {0,1}
    }
\end{equation*}
\caption{Executing $A_\theta$ gate by gate teleportation.}
\label{fig:Gate teleportation}
\end{figure}

\subsection{Quantum one-time pad}
In the quantum one-time pad, a user(Alice) generates 2-bit encryption key $a$,$b$ $\in$ \{0,1\} using coin flips, and executes an $X^a$ gate and a $Z^b$ gate to encrypt her input.
The state after encryption $\ket{\psi}_{\text{enc}}$ is 
\begin{equation*}
    \ket{\psi}_{\text{enc}}=X^aZ^b\ket{\psi},
\end{equation*}
where $\ket{\psi}$ is input, and the $X$ and $Z$ gate gates are Pauli matrices:
\begin{equation*}
X=\begin{pmatrix}0 & 1\\ 1 & 0\end{pmatrix} \qquad  Z=\begin{pmatrix}1 & 0\\ 0 & -1\end{pmatrix}.
\end{equation*}

The density matrix obtained by the server(Bob) receiving this quantum state is maximally mixed state, as described below:
\begin{equation*}
    \frac{1}{4}\sum^{1}_{a,b=0}X^aZ^b\ket{\psi}\bra{\psi}Z^bX^a=\frac{I}{2}.
\end{equation*}
Therefore, Bob, who does not know the randomly generated encryption key, cannot learn the input from the received qubits.

Alice, who has received the quantum state after performing the calculation $U$, can decrypt this state  using the encryption keys $a'$ and $b'$, which were altered by the calculation to:
\begin{equation*}
    Z^{b'}X^{a'}U\ket{\psi}_{\text{enc}}=Z^{b'}X^{a'}(X^{a'}Z^{b'}U\ket{\psi})=U\ket{\psi}.
\end{equation*}
For U it suffices to consider only H, T, and CNOT, which form a universal gate set. 
The correspondence between pairs ($a$, $b$) and ($a'$, $b'$) for each gate is described in Figures \ref{fig:X gate}-\ref{fig:T gate}.  

Note that when the T gate is executed as shown in Figure \ref{fig:T gate}, we obtain the following $S$ gate to be modified in addition to the $X$ and $Z$ gates.  
Here the $S$ gate is
\begin{equation*}
S=T^2=\begin{pmatrix}1 & 0\\ 0 & i\end{pmatrix}.
\end{equation*}
The universal gate set requires a  non-Clifford gate, such as a T or Toffoli gate\cite{Gottesma-Knill}.
When executing these non-Clifford gates, Alice requires modifications other than the $X$ gate and the $Z$ gate\cite{Childs, half}.

\begin{figure}[H]
\begin{equation*}
     \Qcircuit @C=1em @R=.7em {
    X^aZ^b\ket{\psi} & & & &\qw& \qw & \gate{X} & \qw & \qw & & & X^aZ^bX\ket{\psi}
    }
\end{equation*}
\caption{Key change at the $X$ gate.}
\label{fig:X gate}
\end{figure}
\begin{figure}[H]
\begin{equation*}
     \Qcircuit @C=1em @R=.7em {
    X^aZ^b\ket{\psi} & & & &\qw& \qw & \gate{Z} & \qw & \qw & & & X^aZ^bZ\ket{\psi}
    }
\end{equation*}
\caption{Key change at the $Z$ gate.}
\label{fig:Z gate}
\end{figure}
\begin{figure}[H]
\begin{equation*}
     \Qcircuit @C=1em @R=.7em {
    X^aZ^b\ket{\psi} & & & &\qw& \qw & \gate{H} & \qw & \qw & & & X^bZ^aH\ket{\psi}
    }
\end{equation*}
\caption{Key change at the $H$ gate.}
\label{fig:H gate}
\end{figure}
\begin{figure}[H]
\begin{equation*}
    \Qcircuit @C=1em @R=.7em {
    X^aZ^b\ket{\psi}_1 & & & & \qw & \ctrl{1} & \qw & \qw & & & & & & X^aZ^{b\oplus d}CNOT_{\text{control}}\ket{\psi}_1\\
    X^cZ^d\ket{\psi}_2 & & & &  \qw & \targ  & \qw & \qw & & & & & & X^{a\oplus c}Z^dCNOT_{\text{target}}\ket{\psi}_2
    }
\end{equation*}
\caption{Key change at the $CNOT$ gate.}
\label{fig:CNOT gate}
\end{figure}

\begin{figure}[H]
\begin{equation*}
     \Qcircuit @C=1em @R=.7em {
    X^aZ^b\ket{\psi} & & & &\qw& \qw & \gate{T} & \qw & \qw & & & & X^aZ^{a \oplus b}P^aT\ket{\psi}
    }
\end{equation*}
\caption{Key change at the $T$ gate.}
\label{fig:T gate}
\end{figure}

\section{$A_\theta$ gate and universal quantum computation}
In this section, we explain the $A_\theta$ gate that is important for our BQC protocol.
First, we extend the quantum one-time pad by adding the $A_\theta$ gate.
Next, we show how to modify the $T$ gate and the $A_\theta$ gate in the quantum one-time pad.
Finally, we explain universal quantum computation using the $A_\theta$ gate.

\subsection{Quantum one-time pad for the $A_\theta$ gate}
In this subsection we show that the $A_\theta$ gate can be hidden using the quantum one-time pad.
 When $\ket{A_\theta}$ is encrypted using the quantum one-time pad, it is given by
\begin{equation*}
    \ket{A_\theta}_{\text{enc}}=X^aZ^b\ket{A_\theta}.
\end{equation*}
When executing gate teleportation by $\ket{A_\theta}_{\text{enc}}$, the $A_\theta$ gate works as shown in Figure \ref{fig:Atheta}, since the $Z$ gate commutes with the $A_\theta$ gate.
Thus, it is possible to encrypt the $A_\theta$ gate using the quantum one-time pad.
Note that the $A_\theta$ state is in the maximum mixed state without encryption using the $X$ gate, meaning such encryption using the $X$ gate is not required. 
But the $X$ gate is used to hide the measurement result.
It is using for the modification described in the following subsection.

\subsection{Modifying the $A_\theta$ gate}
To apply the $T$ gate and the $A_\theta$ gate to a quantum state encrypted by the $X$ gate, the quantum state requires modification.
When the $A_\theta$ gate is applied to the quantum state encrypted by the $X$ gate, an $A_{-\theta}$ gate is applied to the quantum state instead of the $A_\theta$ gate as shown in Figure \ref{fig:enc Atheta}.
When Alice obtains the undesired measurement results, the $A_{\theta}$ gate is executed, since the 
angle is flipped.
Thus, even if it is encrypted with the $X$ gate, flipping the desired measurement result Alice can execute the $A_\theta$ gate without changing $A_\theta$ state.
The $T$ gate is a special case of the $A_\theta$ gate and can be modified in a similar way.

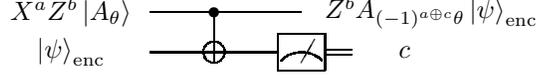
\begin{figure}
\begin{equation*}
     \Qcircuit @C=1em @R=.7em {
    X^aZ^b\ket{A_\theta} & & & & \qw & \ctrl{1} & \qw & \qw & & & & Z^bA_{(-1)^{a\oplus c}\theta}\ket{\psi}_{\text{enc}}\\
    \ket{\psi}_{\text{enc}}  & & & &  \qw & \targ  & \qw & \meter & \cw & &  c
    }
\end{equation*}
\caption{Key change at the $A_\theta$ gate using gate teleportation}
\label{fig:Atheta}
\end{figure}

\begin{figure}
\begin{equation*}
     \Qcircuit @C=1em @R=.7em {
    XZ^b\ket{\psi} & & & &\qw& \qw & \gate{A_\theta} & \qw & \qw & & & & XZ^{b}A_{-\theta}\ket{\psi}
    }
\end{equation*}
\caption{Applying the $A_\theta$ gate to a quantum state encrypted by the quantum one-time pad.}
\label{fig:enc Atheta}
\end{figure}
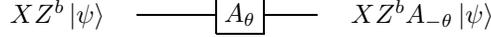

As mentioned above, the $A_\theta$ gate requires additional correction because the angle $\theta$ executed varies depending on the measurement result.
Here, the angle is limited to $\theta = \frac{n \pi}{4}(n = \{0,1, \ldots, 7\})$.
If the measurement gives an undesired result, an $A_{-\theta}$ gate is executed, meaning Alice needs to execute an $A_{2\theta}$ gate to correct it.
The gate used for this correction may also bring  about undesired measurement results.
The next gate for the second modification is the $A_{4\theta}$ gate, for which the angle $\theta$ is limited to $\theta = \frac{n \pi}{4}(n = \{0,1, \ldots, 7 \})$. Thus $A_{4\theta} = Z$ or $I$, meaning the correction can be completed by executing the $Z$ gate or the $I$ gate.
If $\theta$ is limited to $\theta = \frac{n\pi}{4}(n = \{0, 1, \ldots, 7 \})$, Alice can, therefore, be sure to execute the $A_\theta$ gate by preparing two additional qubits and one additional gate.

\subsection{T-like gate group and one-qubit universal gate}
In\cite{Nielsen-Chuang}, an approximation of any one-qubit gate using the $T$ gate and the $H$ gate is achieved, as these two gates can achieve non-parallel two-axis rotation on  the Bloch sphere.
We represent gate blindness using non-parallel eight-axis rotation  and a T-like gate. The T-like gate is defined as follows:

\begin{eqnarray*}
T=A_{\frac{i\pi}{4}}=\begin{pmatrix}1 & 0\\ 0 & e^{\frac{i\pi}{4}}\end{pmatrix},
\qquad T^3=A_{\frac{i3\pi}{4}}=\begin{pmatrix}1 & 0\\ 0 & e^{\frac{i3\pi}{4}}\end{pmatrix},\\
T^\dag=A_{\frac{-i\pi}{4}}=\begin{pmatrix}1 & 0\\ 0 & e^{\frac{-i\pi}{4}}\end{pmatrix},
\qquad (T^3)^\dag=A_{\frac{-i3\pi}{4}}=\begin{pmatrix}1 & 0\\ 0 & e^{\frac{-i3\pi}{4}}\end{pmatrix}.
\end{eqnarray*}

By combining the T-like gate with the $H$ gate, it is possible to make rotations at eight axes that are  not parallel.
Table \ref{tab:Axis of rotation} shows these eight axes along with the gate  combinations.
In particular, note that the rotation axis of $T^\dag HT^\dag H$ is parallel to the rotation axis of $HTHT$.
It is known that an arbitrary one-qubit gate can be approximated by the combination of $THTH$ and $HTHT$\cite{Nielsen-Chuang}.
The rotations of these two axes are not parallel, meaning that they  can achieve any rotation for the quantum state that Alice desires.
This is known thus a universal gate set for a one-qubit gate.
In the same way, any one-qubit gate can be approximated  by the gate group shown in Table \ref{tab:Axis of rotation}.

\begin{table}[h]
  \begin{center}
   \caption{Axes of rotation and the corresponding gate combinations}
   \begin{tabular}{|c|c|} \hline
    Gate & Axis of rotation \\ \hline \hline
    $THTH$ & ($\cos{\frac{\pi}{8}}$,$\sin{\frac{\pi}{8}}$,$\cos{\frac{\pi}{8}}$) \\ \hline
    $THT^\dag H$ & ($-\cos{\frac{\pi}{8}}$,$-\sin{\frac{\pi}{8}}$,$\cos{\frac{\pi}{8}}$) \\ \hline
    $T^\dag HTH$ & ($\cos{\frac{\pi}{8}}$,$-\sin{\frac{\pi}{8}}$,$-\cos{\frac{\pi}{8}}$) \\ \hline
    $T^\dag HT^\dag H$ & ($-\cos{\frac{\pi}{8}}$,$\sin{\frac{\pi}{8}}$,$-\cos{\frac{\pi}{8}}$) \\ \hline
    $T^3HT^3H$ & ($\cos{\frac{3\pi}{8}}$,$\sin{\frac{3\pi}{8}}$,$\cos{\frac{3\pi}{8}}$) \\ \hline
    $T^3H(T^3)^\dag H$ & ($-\cos{\frac{3\pi}{8}}$,$-\sin{\frac{3\pi}{8}}$,$\cos{\frac{3\pi}{8}}$) \\ \hline
    $(T^3)^\dag HT^3H$ & ($\cos{\frac{3\pi}{8}}$,$-\sin{\frac{3\pi}{8}}$,$-\cos{\frac{3\pi}{8}}$) \\ \hline
    $(T^3)^\dag H(T^3)^\dag H$ & ($-\cos{\frac{3\pi}{8}}$,$\sin{\frac{3\pi}{8}}$,$-\cos{\frac{3\pi}{8}}$) \\ \hline
   \end{tabular}
   \label{tab:Axis of rotation}
  \end{center}
\end{table}

Bob cannot discover which gate combination was chosen because he cannot obtain the received state $\ket{A_\theta}$.
At the same time, Alice can realize any one-qubit gate without it being known to Bob.
We make use of this feature for the BQC protocol.

\section{Main protocol}
In this section, we describe two protocols.
In Protocol 1, Alice performs the calculation without Bob  knowing the input/output or the calculation process other than position of the $CNOT$ gate. 
In Protocol 2, Alice conceals the input/output and calculation process.

In the following, Bob has a universal quantum computer, and Alice has the ability to prepare a computational basis $\ket{0}, \ket{1}$, and a state $\ket {A_\theta}$ such that $\theta = \frac{n\pi}{4} (n=\{0,1,\ldots ,7\})$ to execute the $X$ and $Z$ gates and perform classical calculations.
It should be  noted that this protocol does not require Bob to have the ability to execute the non-Clifford  gate group.

\subsection{Protocol 1}
According to Section 3, we can execute any one-qubit gate without it being known to Bob. 
Our first protocol uses such a hidden one-qubit gate.

\begin{description}
\setlength{\parskip}{3mm}
    \item[Step l.] Alice makes a calculation circuit for her calculation.
    
    \item[Step 2.] Alice converts the circuit to a weak blind circuit according to Table \ref{tab:Axis of rotation}.
    This  conversion is optional, meaning Alice can create a number of structure circuits. 
    Alice chooses one of them.
    
    \item[Step 3.] Alice encrypts the necessary input qubits using the quantum one-time pad and sends them  to Bob.
    In addition, Alice encrypts the ancilla bits required for the $T$ gate and its modification with the  quantum one-time pad and sends them to Bob.

   \item[Step 4.] After sending all the qubits, Alice sends Bob the circuits for the computation.
   Bob computes the circuits using the ancilla bits as well as his $H$ and $CNOT$ gates.
   Then,  Bob sends the measurement result to Alice and asks whether it is the desired result.
   If the result is not the  desired one, Bob makes additional modifications using additional ancilla bits and his $Z$ gate.
   
   \item[Step 5.] Bob sends the qubits to Alice after the calculation is completed. 
   Alice unencrypts the sendsent qubits, measures them, and obtains the result.
\end{description}
Alice can perform the calculation without Bob knowing the input/output or the calculation  process aside from the position of the $CNOT$ gate. 
Here, we define weak blindness.

\begin{defn}[weak blindness]
Let P be a quantum delegated computation on input X and let L(X) be any function
of the input. We say that a quantum delegated computation protocol is weak blind while leaking at most
L(X) and position of the $CNOT$ gate if, on Alice's input X, for any fixed Y = L(X), the following two hold when given Y :
\begin{itemize}
    \item[1.] The distribution of the classical information obtained by Bob in P is independent of X.
    \item[2.]Given the distribution of classical information described in 1, Bob cannot know about 1-qubit gate executed between the $CNOT$ gate.
\end{itemize}
\end{defn}
\setcounter{thm}{0}
\begin{thm}
Protocol 1 is weak blind while leaking at circuit size and the $CNOT$ gate position.
\end{thm}
\begin{proof}
Bob obtains information on the circuit size and CNOT position based on Alice's calculation procedure, while Alice's input and ancilla bits are encrypted by the quantum one-time pad.
The  encryption key does not depend on the input and ancilla bits, meaning Bob knows nothing about them.
Bob makes a measurement when executing a one-qubit gate, but the  measurement result has a success probability of 1/2, regardless of the gate executed and the input.
Therefore, Bob does not obtain any information when executing the one-qubit gate.
Bob cannot learn anything about the output because the computed state, which is  the output, is still encrypted by the quantum one-time pad.
Protocol 1 satisfies weak blindness because Bob does not know anything other than the circuit size and the CNOT position.
\end{proof}

Bob does not know anything about the state of the qubits he receives, which Alice has encrypted using the quantum one-time pad. 
At the same time, the quantum operation and  measurement result do not depend on the contents of the input and the one-qubit gate.
Thus, Bob can determine only the size of the input and the position of the CNOT gate. 
Note also that  Alice can ncrease the input size by sending dummy ancilla bits.

However, as Bob knows the position of the CNOT gate, there is athe possibility that he can infer  the algorithm based on this information when executing a known algorithm,  such as Shor’s or Grover’s algorithm.
However, if an algorithm or application is unknown to the  public, Alice can execute it without Bob being able to  obtain knowledge of it. 
In other words, unpublished algorithms and applications can be tested without fear  of Bob eavesdropping.

Though the existing version of Broadbent's protocol \cite{Broadbent} requires one ancilla bit per $T$ gate, the ability to conceal unpublished algorithms can be achieved by adding two additional  ancilla bits to the protocol.
However, in plotocol 1 Alice's ability must be higher than that required by Broadbent's protocol; she must have the ability to prepare $\ket{0} and \ket{1}$ and execute the $X$, $Z$, $H$, and $S$ gates.

\subsection{Protocol 2}

\begin{figure}[t]
\centering
\includegraphics[scale=0.5]{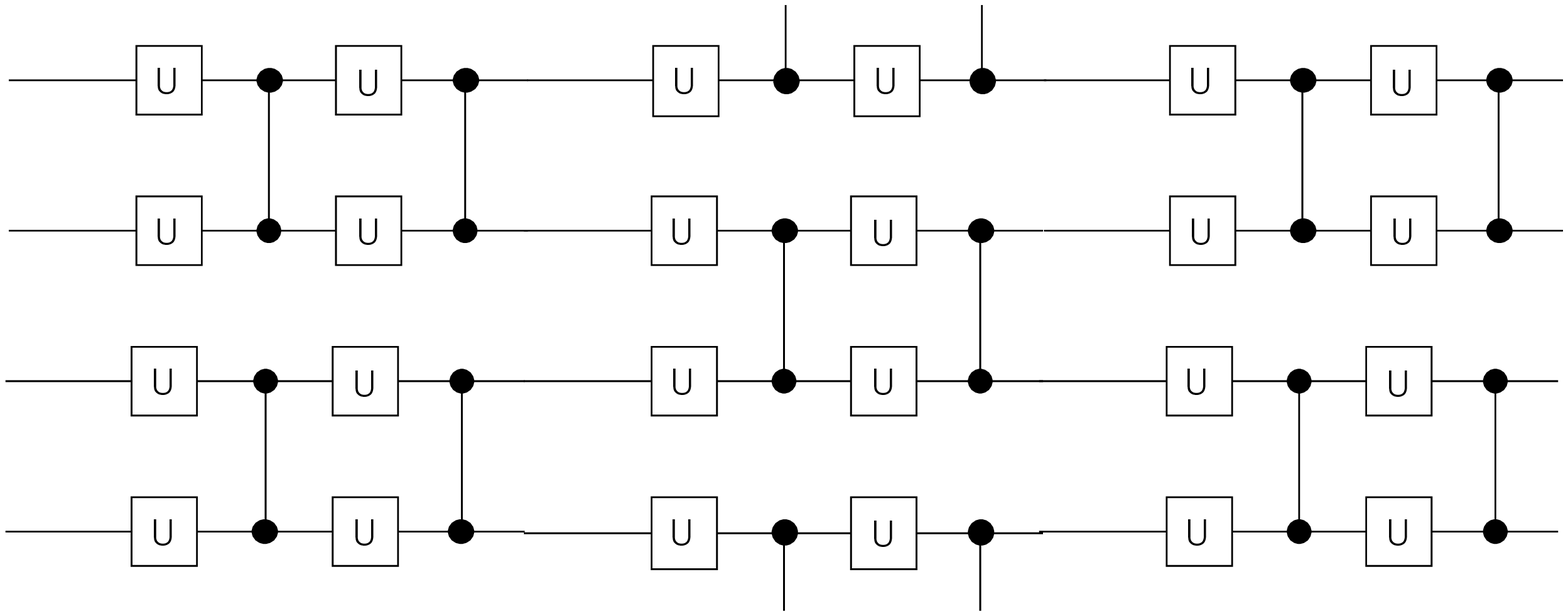}
\caption{Circuit based on BFK protocol graph}
\label{fig:BFK}
\end{figure}

A circuit referring to the BFK protocol can be created, such as the one shown in Figure \ref{fig:BFK}. 
The  identity gate and the $CNOT$ gate can be realized by combining the two $CZ$ gates and one-qubit gates, as shown in Figures \ref{fig:make identity} and \ref{fig:make CNOT}.
The identity gate and the $CNOT$ gate can be realized by combining the two $CZ$ gates and one-qubit gates as shown in Figure \ref{fig:make identity}-\ref{fig:make CNOT}.
Here, $R_z(\frac{\pi}{4})$ and $R_x(\frac{\pi}{4})$ represent the z-axis and x-axis rotation of the Bloch sphere, respectively.
Since Bob cannot obtain information concerning one-qubit gates, he cannot determine whether the identity gate or  the CNOT gate is realized by the two $CZ$ gates.
By arranging the two sets of the $CZ$ gates alternately and in a staggered manner, as shown in Figure \ref{fig:BFK}, Alice can perform the calculation without letting Bob know the position of the $CNOT$ gate.

\begin{figure}[t]
\centering
\includegraphics[scale=0.5]{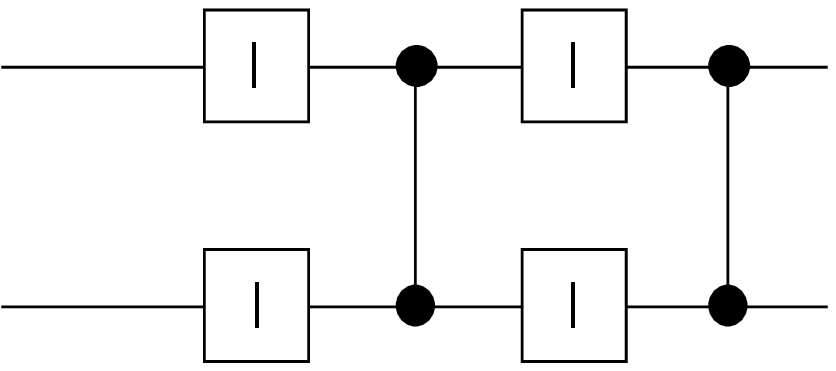}
\caption{Combination of the $CZ$ gates and one-qubit gates acting as the identity gate}
\label{fig:make identity}
\end{figure}

\begin{figure}[t]
\centering
\includegraphics[scale=0.5]{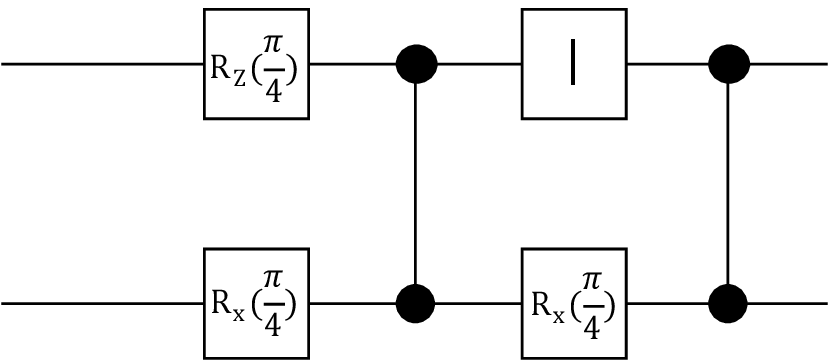}
\caption{Combination of the $CZ$ gates and one-qubit gates acting as the $CNOT$ gate}
\label{fig:make CNOT}
\end{figure}

The procedure for Protocol 2 is that same as the procedure for Protocol 1.
However, it is here necessary to send extra ancilla bits for single qubits between the $CZ$ gates.

\begin{description}
\setlength{\parskip}{3mm}
    \item[Step 1.] Alice makes a calculation circuit for her calculation.
    
    \item[Step 2.] Alice converts the circuit to a blind circuit using Table \ref{tab:Axis of rotation} and Figures \ref{fig:BFK}-\ref{fig:make CNOT}. 
    Note that number of one-qubit gates between two $CZ$ gates is constant so as to achieve blindness.
    When the number of one-qubit gates is not constant, Alice adds dummy qubits.
    This conversion is optional, meaning that Alice can create a number of structure circuits.
    Alice chooses one of them.
    
    \item[Step 3.] Alice encrypts the necessary input the qubits using quantum one-time pad and sends them to Bob.
    In addition, Alice encrypts the ancilla bits required for the $T$ gate and its modification with the quantum one-time pad and sends them to Bob.

   \item[Step 4.] After sending all the qubits, Alice sends Bob the circuits for the computation.
   Bob computes the circuits using the ancilla bits as well as his $H$ and $CZ$ gates.
   At this time, Bob sends the measurement result to Alice and asks whether it is the desired result.
   If the result is not the one desired, Bob makes additional modifications using additional ancilla bits and his $Z$ gate.
   
   \item[Step 5.] Bob sends the qubits to Alice after the calculation is completed.
   Alice unencrypts the sent qubits, measures them, and obtains the result. 
\end{description}

Alice can perform quantum computation while concealing the entire  calculation process, including the position of the $CNOT$ gate.
Here, we define blindness.

\begin{defn}[Blindness{\cite[Definition 2]{BFK}}]
Let P be a quantum delegated computation on input X and let L(X) be any function
of the input. We say that a quantum delegated computation protocol is blind while leaking at most
L(X) if, on Alice's input X, for any fixed Y = L(X), the following two hold when given Y :
\begin{itemize}
    \item[1.] The distribution of the classical information obtained by Bob in P is independent of X.
    \item[2.]Given the distribution of classical information described in 1, the state of the quantum system
obtained by Bob in P is fixed and independent of X.
\end{itemize}
\end{defn}
\setcounter{thm}{1}
\begin{thm}
Protocol 2 is blind while leaking at circuit size.
\end{thm}
\begin{proof}
Bob obtains information on the circuit size based on Alice's calculation procedure.
Alice's input and ancilla bits are encrypted by the quantum one-time pad, and the encryption key does not depend on the input and ancilla bits, meaning Bob knows nothing about them.
Bob takes his measurement while executing a one-qubit gate, but the measurement result has a success  probability of 1/2 regardless of the gate executed and the input.
Therefore, Bob  does not obtain any information when executing the one-qubit gate.
The CNOT gate cannot be  distinguished from the identity gate by combining one-qubit gates and CZ gates, so Bob can  execute the CNOT gate without knowing where in the circuit it was executed.
Bob cannot  learn anything about the output because the computed state, which is the output, is still encrypted by  the quantum one-time pad. 
Protocol 2 satisfies blindness because Bob does not know anything  other than the circuit size.
\end{proof}

Bob can obtain no information about the state of the qubit received from Alice and encrypted by the  quantum one-time pad, and the quantum operation and measurement result do not depend on the input and  calculation processes.
Therefore, Bob can know only the size of the input, and Protocol 2  satisfysatisfies blindness.
Again, note that the size can be increased by  sending dummy ancilla bits.

In protocols developed in previous research, the blindness of all calculation processes was  performed simultaneously by Bob for each gate, including \{H, P, T, CZ, CNOT \}, from a  necessary qubit to computation and dummy qubits.
Then, Bob sent those qubits back to Alice,  who saved the necessary qubits in quantum memory and sent them back when Bob needed to execute the necessary gates\cite{Fitzsimons review, Full}.
According to Protocol 2, the quantum memory and  additional quantum communication can be reduced, and Alice's ability  needs only to be equivalent to that required by the BFK protocol. 
However, this protocol requires more  ancilla bits than Protocol 1.

\section{Verification and fault tolerance}

\subsection{Verification}
Verification is closely related to blindness\cite{BFK}.
If the problem that Alice wishes to solve is included in the computational complexity class NP, it can be verified using a classical computer, but it is also believed that BQP is BQP $\not\subset$ NP \cite {NP-BQP1, NP-BQP2}.
Therefore, it is difficult to verify whether the result of the problem included in BQP is correct using a classical computer and it is also difficult for Alice to do so, for she has few quantum resources.
Thus, it is necessary to verify that the evil Bob does not follow Alice's particular instructions and instead performs different operations.
Here, we show that verification using trap qubits is possible in both Protocol 1 and  Protocol 2. 
Note that a method that does not use trap qubits, which can be used for blindness calculations, is  also known\cite {verification}. 
Even if Alice has only a classical computer, there is still a method that can doperform verification\cite{verification2}.

The following verification method using trap qubits can be used for both Protocol 1 and Protocol 2.
Since the input and the one-qubit gate are hidden in these protocols, the trap qubit can be put in the input qubit.
The trap qubit $\ket{0}$ or $\ket{+}$ is encrypted with quantum one-time pad, and all trap qubit gates are implemented as the identity gates, which are the $A_\theta$ gates with $\theta = 0$.
However, since the identity gate succeeds with probability 1, there must be a 1/4 chance that the desired measurement result is not obtained once (Alice executes the $A_\theta$ gate operation with two $\theta = 0$ qubits, which are identity gates).
At the same time, there must also be a 1/4 chance that the desired measurement result is not  obtained twice (Alice executes the Ag gate, which realizes the identity gate of $\theta = 0$ and the $Z$ gate of $\theta = \pi$, and lets Bob correct it with the $Z$ gate).
When the evil Bob attempts to operate the gate differently from Alice's instruction, Alice can know stochastically whether or not he operates on the trap qubit.
In the case of calculating N qubits mixed with $N_d$ trap qubits, Alice can detect evil Bob's operation with the probability of $\frac{N_d}{N}$.
Alice can increase the probability of detection to $1-(\frac{N-N_d}{N})^s$ by performing the same calculation s times.

\subsection{Fault tolerant quantum computation}
It is known that the ability to perform error correction in a universal quantum computer is an indispensable function, since coherence is destroyed by external noise when manipulating a  quantum state \cite{Nielsen-Chuang,fault1,fault2,fault3,fault4}.
It has also been shown that there is no universal gate set that is transversal  (does not spread errors) \cite{trans-univ1,trans-univ2}.
However, it is known that the H gates and the CNOT gates can implement error correction codes in a transversal manner (without spreading errors)\cite{falt t gate1,falt t gate2}.
For the $T$ gate, this method is implemented only by transversal CNOT gates and measurement by teleportation.
In this protocol, the gates used in Bob's calculation are only the $H$ gate and the $CNOT$ gate, and a non-transversal T-like gate can execute a logical T-like gate by preparing multiple similar $A_\theta$ state.
Therefore, It can be extended to fault tolerant calculations.
In the proposed protocols, only the $H$ gate and the $CNOT$ gate are used in  Bob's calculation, and a non-transversal T-like gate can execute a logical T-like gate by preparing multiple  similar $A_\theta$ statestates. 
Therefore, the protocols can be extended to fault -tolerant calculations.

\section{Conclusion}
In this paper, we proposed two BQC protocols using circuit -based quantum computation (CBQC).
These two protocols execute computation with weak blindness and blindness, respectively.
In previous research\cite{Childs,Broadbent}, it has been discovered that Alice's input and output can be concealed from Bob using the quantum one-time pad. 
However, previous  techniques did not conceal the calculation process.
In our protocol, blindness was achieved using gate teleportation and expanding the $T$ gate, which is important for universal quantum computation, to a T-like gate.
First, we proposed a protocol
with weak blindness (Protocol 1),  which discloses the position of the CNOT gate in addition to the size of the circuit, for which a non-disclosed calculation algorithm is sufficient.
In Protocol 2, BQC  was achieved using CBQC.
Further, it was these were shown that verification using trap qubits is possible using these protocols.
We also showed that the method can be extended to fault-tolerant calculations by the same method for error correction using magic state.

\section*{Acknowledgment}
We would like to thank Takayuki Miyadera for many helpful comments, and we are grateful to Ikko Hamamura for important advice regarding the protocols.

\end{document}